\documentclass[]{article}
\usepackage{amsmath,amsthm}
\usepackage{amsfonts}
\usepackage{amssymb}

\theoremstyle{definition}
\newtheorem{defn}{Definition}[section]

\theoremstyle{definition}
\newtheorem{lemma}{Lemma}[section]

\theoremstyle{definition}
\newtheorem{theorem}{Theorem}[section]

\theoremstyle{definition}
\newtheorem{corollary}{Corollary}[section]

\theoremstyle{definition}

\theoremstyle{definition}

\theoremstyle{definition}
\newtheorem{remark}{Remark}[section]

\theoremstyle{definition}

\title{Two theorems about the P versus NP problem}
\author{Tianheng Tsui}

\begin{document}

\maketitle

\begin{abstract}
Two theorems about the P versus NP problem be proved in this article (1) There exists a language $L$, that the statement $L \in \textbf{P}$ is independent of ZFC. (2) There exists a language $L \in \textbf{NP}$, for any polynomial time deterministic Turing machine $M$, we cannot prove $L$ is decidable on $M$.
 
\end{abstract}

\begin{center}
	\item
	\section{Introduction}
\end{center}

First of all, the results of this article neither imply  $\textbf{P} = \textbf{NP}$ nor $\textbf{P} \ne \textbf{NP}$. Let me to explain.

The meaning of the first theorem is that we can define a language $L$ formally in ZFC, and there are two models of ZFC, model A and model B, the sentence $L\in \textbf{P}$ is true in model A and false in model B. Therefore we cannot prove or disprove $L\in \textbf{P}$ in ZFC. The method for proving this result is to reduce the independence of $L\in \textbf{P}$ to a sentence $\varphi$ which is already known independent of ZFC.

More specifically, first,we can construct a function $f$ with the property `` $\forall n f(n) =1$ is independent of ZFC'', that is, there are two model of ZFC, model A and model B, the sentence $\forall n f(n) =1$ is true in model A and false in model B. Then construct a Turing machine $M$ depending on $f$, and define the language $L =\{n|\ M(n)=1\}$. Finally, we prove the function $f$ like a switch, such that $L\in \textbf{P}$ is true in model A where $\forall n f(n) =1$ is true, and $L \in \textbf{P}$ is false in model B where $\forall n f(n) =1$ is false, therefore the independence of $L \in \textbf{P}$ be proved.

The meaning of the second theorem is that we can define a language $L$ formally in ZFC, and we can prove $L \in \textbf{P}$ in ZFC, \textbf{but} for any concrete polynomial time Turing machine $M$, we cannot prove $M$ decide $L$ in ZFC. \textbf{The statement is highly confusing, and what does this mean?!}

Let me to explain. 

First, it can be proved in ZFC that there exist a language $L$ and a non-ploynomial time Turing machine $M$ which decide the language $L$, that is $$L =\{n|\ M(n)=1\} \text{ and } M \not\in \textbf{P}$$.

Second, in order to prove $L \in \textbf{P}$ in ZFC, there are actually two methods, constructive proof and non-constructive proof:\\

\begin{enumerate}
	\item Constructive method: Proving a concrete ploynomial time Turing machine $M_1$ satisfies $\forall n (M(n) = M_1(n))$ in ZFC, i.e., we must construct a concrete Turing machine $M_1$ and give a proof of the sentence $$(M_1 \in \textbf{P}) \wedge \forall n (M(n) = M_1(n))$$ in ZFC, and more, from the soundness of first order logic, we know $M_1$ statisfies the conditions in every model of ZFC. \textbf{But} we will see that such constructive method does not exist, because following reason:
	\item The second method is non-constructive. It prove in each model of ZFC, $W$, there exist a corresponding  concrete ploynomial time Turing machine $M_W$, the sentence $(M_W \in \textbf{P}) \wedge \forall n (M(n) = M_W(n))$ is true in model $W$, that is, in each model of ZFC, $L \in \textbf{P}$ is true, therefore from the G\"odel's completeness theorem, we can prove $L \in \textbf{P}$ in ZFC. \textbf{But} we will prove the $M_A$ in model A must be different from the $M_B$ in model B, where the model A and model B are two different models of ZFC as in the first theorem of this article. Hence we cannot use the first constructive method to prove $L \in \textbf{P}$ in ZFC. That means for any concrete polynomial time Turing machine $M$, we cannot prove $M$ decide $L$ in ZFC.     
\end{enumerate}

Since $\textbf{P} \subset \textbf{NP}$ and above explanation, we get the second theorem: there exists a language $L \in \textbf{NP}$, for any polynomial time deterministic Turing machine $M$, we cannot prove $L$ is decidable on $M$.

\begin{center}
	\item
	\section{Preliminaries and Notation}
\end{center}

This section is devoted to the exposition of basic preliminary material, notations and conventions which be  used throughout of this article.

The notion of algorithm can be defined in terms of Turing machines by means of the Church–Turing thesis, so the sentence ``There is an algorithm \ldots '' means ``There is a Turing machine to compute \ldots '', and sometimes Turing machine algorithms be described in very high level. If a function or a map is recursive, it means that the function or the map can be computed by a Turing machine.

Let the formal Zermelo-Fraenkel axiomatic set theory is denoted by ZF, and ZFC denotes the theory ZF with the Axiom of Choice, and $\mathcal{\omega}$ represents the natural number set $N$ in the formal ZFC system.

\begin{defn}
	\label{defn_classical_tm}
	A Turing machine $\textbf{M}$ is a 5-tuple, $(Q,\Gamma,\delta,q_0,q_{halt})$
	
	$Q$ is a finite set of states, i.e., $\exists i(i \in \omega \wedge \Vert Q \Vert = i)$,
	
	$\Gamma$ is the tape alphabet containing the blank symbol $\sqcup$, and the left end symbol $\triangleright$,
	
	$\delta$:$Q \times \Gamma \longrightarrow Q \times \Gamma \times \{L,S,R\}$ is the transition function,
	
	if $\delta(q,\triangleright) = (p,s,b)$, then $(s = \triangleright) \wedge (b = R)$,
	
	if $\delta(q, a) = (p,s,b)$ and $a \ne \triangleright$, then $s \ne \triangleright $,
	
	$q_0 \in Q$ is the start state,
	
	$q_{halt} \in Q$ is the halt state, that is $\forall a \in \Gamma$:
	
	$\delta(q_{halt}, a) = (q_{halt},a,S)$, and
	
	$\forall q \in Q (\forall a \in \Gamma (\delta(q, a) = \delta(q_{halt}, a)) \rightarrow (q=q_{halt}) )$.

\end{defn}

	Unless otherwise indicated, it will always be assumed that the tape alphabet $\Gamma = \{0,1, \sqcup, \triangleright \}$ throughout this article, and we assume the basic notions and results of mathematical logic, such as formula, sentence, the set of all formulas is recursive \ldots, etc.


\begin{defn}
	\label{time_complexity}
	(\textbf{time complexity})  Let $M$ be a Turing machine that halts on all inputs. The running time or time complexity of $M$, denoted by $t_M$, is the function $$t_M:\ \omega \rightarrow \omega$$ where $t_M(n)$ is the maximum number of steps that $M$ uses on any input of length $n$. 
\end{defn}

\begin{center}
	\item
	\section{A theorem about P class problems}
\end{center}

Let the formal Zermelo-Fraenkel axiomatic set theory is denoted as ZF, and ZFC denotes the theory ZF with the Axiom of Choice, and $\mathcal{\omega}$ represents the natural number set $N$ in the formal ZFC system. In this section, a theorem about $ \textbf{P}$  class problems be proved: There exists a language $L$, that the statement $L \in \textbf{P}$ is independent of ZFC.

Let $D$ denotes the set of all deducible(provable) expressions of ZFC. Suppose ZFC is a consistent effective formal system, then $D$ is enumerable\cite{yim}. It is obvious that the concepts and statements such as ``Turing machine'', ``Turing machine $M$ do not halts on input $w$'', etc$\ldots$, can be expressed in ZFC.\\  

\begin{lemma}
	\label{lemma_001}
	There exists a Turing machine $M_h$ and an input $w_h$, and let $\varphi$ denotes the statement ``$M_h$ halts on input $w_h$'', then:

\begin{enumerate}
	\item $\varphi$ cannot be proved in ZFC.
	\item $\neg \varphi$ cannot be proved in ZFC.
	\item There exist a model of ZFC, $\mathcal{A}$, $\varphi$ is ture in it.
	\item There exist a model of ZFC, $\mathcal{B}$, $\neg \varphi$ is ture in it.
\end{enumerate}	
hence $\varphi$ is independent of ZFC.

\end{lemma}

\begin{proof}
	Let
	\[
	HALT_\textbf{TM} = \{ \langle M, w \rangle |  M \text{ is a Turing machine and }M \text{ halts on input }w \}
	\]

	Let Turing machine $U$ as:\\$U$ = ``On input $\langle M, w \rangle$, where $M$  is a Turing machine and  $w$ is a string:

	\begin{enumerate}
		\item Simulate  $M$  on input  $w$.
		\item At the same time, enumerate the deducible expressions from ZFC, and find whether the expression is ``$M$ do not halts on input $w$'', one by one.
		\item If $M$ ever enters its halt states on input $w$, $U$ accept $\langle M, w \rangle$ and halt.
		\item If found a deducible expression is ``$M$ do not halts on input $w$'', $U$ reject $\langle M, w \rangle$ and halt.''
	\end{enumerate}
	It is well known that $HALT_\textbf{TM}$ is undecidable\cite{hl,ms}. So $U$ does not decide  $HALT_\textbf{TM}$, therefore there exist $\langle M_h, w_h \rangle$, $U$ never halts on it. That means:
	
	\begin{enumerate}
		\item From the third item of the above definition of $U$ and the fact that $U$ never halts on $\langle M_h, w_h \rangle$, we can conclude Turing machine $M_h$ never halts on the input $w_h$, so the statement: ``$M_h$ halts on input $w_h$'', is not a deducible expression of ZFC, i.e., $\varphi$ cannot be proved in ZFC.
		\item From the fourth item of the above dedinition of $U$ and the fact that $U$ never halts on $\langle M_h, w_h \rangle$, the statement: ``$M_h$ does not halt on input $w_h$'' , is not a deducible expression of ZFC, i.e.,  $\neg \varphi$ cannot be proved in ZFC.
	\end{enumerate}

	From the G\"odel's completeness theorem\cite{dm} that means there is a model of ZFC, $\mathcal{A}$ , $\varphi$ is ture in it, and there is another model of ZFC, $\mathcal{B}$,
	$\neg \varphi$ is ture in $\mathcal{B}$.

\end{proof}

Let strings in $\{0,1\}^*$ be used to represent the nonnegative integers in the familiar binary notation.

\begin{defn}
	
	$s = a_1 a_2\ldots a_n \in \{0,1\}^*$, let $\|s\|$ denotes the length of the string $s$, (i.e., $\|s\| = n$) and
	
	$\textbf{num}(s) = a_1 \cdot 2^{n-1} + a_2 \cdot 2^{n-2} + \ldots + a_n .$

\end{defn}

\begin{defn}
	\label{oms}
	For any Turing machine $M$ and an input $w$, we can define a Turing machine $O_{\langle M, w \rangle}$ on $\{0,1\}^*$ as:\\$O_{\langle M, w \rangle}$ = ``On input $s = a_1 a_2\ldots a_n \in \{0,1\}^*$:
	
	\begin{enumerate}
		\item Compute the length of the string $s$: n = $\|s\|$.
		\item Simulate  $M$  on input  $w$.
		\item if $M$ halts within $n$ steps, return 0, halt; else return 1, halt .''
	\end{enumerate}
\end{defn}

\begin{lemma}
	\label{lemma_002}
	For all $O_{\langle M, w \rangle}$ defined as definition \ref{oms}, the following statements can be proved in ZFC:
	\begin{enumerate}
		\item  \label{lemma_002_01} $O_{\langle M, w \rangle} \in \textbf{P}$
		\item  \label{lemma_002_02} $  (\forall s (s \in \{0,1\}^*) \rightarrow ( O_{\langle M, w \rangle}(s) = 1) ) $, if and only if ``$M$ does not halt on input $w$''.
		\item  \label{lemma_002_03} $  (\exists m(\|s\| < m \rightarrow O_{\langle M, w \rangle}(s)=1 ) \wedge (\|s\| \ge m \rightarrow O_{\langle M, w \rangle}(s)=0) ) $, if and only if ``$M$ halts at the $m$th step on the input $w$''.
		\item  \label{lemma_002_04} $ (\exists m(\|s\| < m \rightarrow O_{\langle M, w \rangle}(s)=1 ) \wedge (\|s\| \ge m \rightarrow O_{\langle M, w \rangle}(s)=0) ) \vee (\forall s  ( O_{\langle M, w \rangle}(s) = 1) ) $
		\item  \label{lemma_002_05} $\|s\|> \|r\| \rightarrow ( O_{\langle M, w \rangle}(s)=1 \rightarrow  O_{\langle M, w \rangle}(r)=1  )$
		\item  \label{lemma_002_06} $\|s\| = \|r\| \rightarrow (  O_{\langle M, w \rangle}(s) = O_{\langle M, w \rangle}(r) ) $
	\end{enumerate}
\end{lemma}

\begin{proof}
	It is obviously from the definition of $O_{\langle M, w \rangle}$.
\end{proof}

From the lemma \ref{lemma_002}, it is easy to get:

\begin{corollary}
	\label{cor_001}
	For all $\langle M, w \rangle$, where $M$  is a Turing machine and  $w$ is a input string:
	\begin{enumerate}
		\item \label{cor_001_01} The statement: ``$M$ does not halt on input $w$'', can be proved in ZFC, if and only if  $  (\forall s (s \in \{0,1\}^*) \rightarrow ( O_{\langle M, w \rangle}(s) = 1) ) $ can be proved in ZFC.
		\item \label{cor_001_02} The statement: ``$M$ halts on input $w$'', can be proved in ZFC, if and only if  $  (\exists m(\|s\| < m \rightarrow O_{\langle M, w \rangle}(s)=1 ) \wedge (\|s\| \ge m \rightarrow O_{\langle M, w \rangle}(s)=0) ) $ can be proved in ZFC.
		\item \label{cor_001_03} The statement: ``$M$ halts on input $w$, or, $M$ does not halt on input $w$'' can be proved in ZFC.
	\end{enumerate}
	
\end{corollary}

\begin{proof}
	It is easy to see:
	
	\begin{center}
		The statement \textbf{\ref{cor_001_01}} corresponds to the statement \textbf{\ref{lemma_002_02}} of lemma \ref{lemma_002}.
		
		The statement \textbf{\ref{cor_001_02}} corresponds to the statement \textbf{\ref{lemma_002_03}} of lemma \ref{lemma_002}.
		
		The statement \textbf{\ref{cor_001_03}} is a $\phi \vee \neg \phi $ form, so it is ture, and corresponds to the statement \textbf{\ref{lemma_002_04}} of lemma \ref{lemma_002}.
	\end{center}
	
\end{proof}

\begin{lemma}
	\label{lemma_003}
	There exists a Turing machine $M_h$ and an input $w_h$, let ${\phi}_1$ denotes the statement $  (\forall s (s \in \{0,1\}^*) \rightarrow ( O_{\langle M_h, w_h \rangle}(s) = 1) ) $, and 
	${\phi}_2$ denotes the statement $  (\exists m(\|s\| < m \rightarrow O_{\langle M_h, w_h \rangle}(s)=1 ) \wedge (\|s\| \ge m \rightarrow O_{\langle M_h, w_h \rangle}(s)=0) ) $ then
	
	\begin{enumerate}
		\item ${\phi}_1 \vee {\phi}_2$ can be proved in ZFC,
		\item ${\phi}_1$ cannot be proved in ZFC,
		\item ${\phi}_2$ cannot be proved in ZFC,
		\item ${\phi}_1$ cannot be disproved in ZFC,
		\item ${\phi}_2$ cannot be disproved in ZFC,
		\item There exist a model of ZFC, $\mathcal{A}$, in which ${\phi}_1$ is true,
		\item There exist a model of ZFC, $\mathcal{B}$, in which ${\phi}_2$ is true.
		
	\end{enumerate}
\end{lemma}

\begin{proof}
	From the \textbf{definition} \ref{oms}, it is easy to see ${\phi}_1 \vee {\phi}_2$ can be proved in ZFC.
	From the lemma \ref{lemma_001}, there exists a Turing machine $M_h$ and an input $w_h$:  
	
	\begin{enumerate}
		\item ``$M_h$ does not halt on input $w_h$'' cannot be proved in ZFC.
		\item  ``$M_h$ halts on input $w_h$'' cannot be proved in ZFC.
		\item There exist a model of ZFC, $\mathcal{A}$, ``$M_h$ does not halt on input $w_h$'' is ture in it.
		\item There exist a model of ZFC, $\mathcal{B}$, ``$M_h$ halts on input $w_h$'' is ture in it.
	\end{enumerate}

	Then from the corollary\ref{cor_001}, the corresponding formulas cannot be proved in ZFC, that is ${\phi}_1$ cannot be proved in ZFC, and ${\phi}_2$ cannot be proved in ZFC. From the definition of     $ O_{\langle M_h, w_h \rangle}(s)$, it is easy to that $ \neg {\phi}_1 \leftrightarrow {\phi}_2 $. Hence ${\phi}_1$, ${\phi}_2$ cannot be disproved in ZFC.
	
	From the lemma \ref{lemma_002},  ``$M_h$ does not halt on input $w_h$'' is ture in  $\mathcal{A}$ deduce ${\phi}_1$ is true in $\mathcal{A}$; ``$M_h$ halts on input $w_h$'' is ture in $\mathcal{B}$ deduce ${\phi}_2$ is true in $\mathcal{B}$

\end{proof}

Let abbreviate the formula $ (\forall s (s \in \{0,1\}^*) \rightarrow ( M(s) = 1) ) $ to $ \forall s  ( M(s) = 1) $.

\begin{corollary}
	\label{cor_of_lemma_03}
	There exists a Turing machine $M$ that it halts on every input and the following five formulas can be proved in ZFC:
	\begin{equation}
		\forall r,s \in \{0,1\}^*\   (\ \|s\| = \|r\| \rightarrow (  M(s) = M(r) )\ )
	\end{equation}
	
	\begin{equation}
		\forall r,s \in \{0,1\}^*\   (\ \|s\|> \|r\|) \rightarrow ( M(s)=1 \rightarrow  M(r)=1  )
	\end{equation}
	
	\begin{equation}
	\forall r,s \in \{0,1\}^*\   (\ \|s\|> \|r\|) \rightarrow ( M(r)=0 \rightarrow  M(s)=0  )
	\end{equation}

	\begin{equation}
	\forall s \in \{0,1\}^*\   (\  M(s)=0 \vee  M(s)=1  )
	\end{equation}
	
	\begin{equation}
	(\forall s (s \in \{0,1\}^*) \rightarrow ( M(s) = 1) ) \vee (\exists m(\|s\| < m \rightarrow M(s)=1 ) \wedge (\|s\| \ge m \rightarrow M(s)=0) )
	\end{equation}

	but the formula $ \forall s  ( M(s) = 1) $ is independent of ZFC, i.e., it cannot be proved in ZFC and its negation is also unprovable in ZFC.
\end{corollary}

\begin{proof}
	Let the machine $M$ is the $O_{\langle M_h, w_h \rangle}$ as mentioned in \textbf{lemma} \ref{lemma_003}. Hence it is obvious.
\end{proof}

\begin{lemma}
	\label{lemma_004}
	Let $L_1 \subseteq \{0,1\}^*, L_2 \subseteq \{0,1\}^* $. \\If  $  \exists m \forall s (s \in \{0,1\}^*) \rightarrow (\|s\| \ge m \rightarrow (s \in L_1 \leftrightarrow s \in L_2 ) ) $,\\then  $L_1 \in \textbf{P} \rightarrow L_2 \in \textbf{P}$.
\end{lemma}

\begin{proof}
	This lemma says that if there exist a number $m$, any string $s  \in \{0,1\}^*$ with $\|s\| \ge m$ imply: ``$s \in L_1$ if and only if $s \in L_2$'', then from the $L_1 \in \textbf{P}$, we can deduce  $L_2 \in \textbf{P}$. 
	
	Let $M$ is the polynomial time deterministic Turing machine which decide $L_1$. Obviously, from the lemma's condition, if $\|s\| \ge m$, $M$ can be used to decide whether $s \in L_2$, else if $ \|s\| < m $, seaching table method can be used to decide whether $s \in L_2$. Hence $L_2 \in \textbf{P}$.

\end{proof}

\begin{defn}
	\label{apeq}
	Let $L_1 \subseteq \{0,1\}^*, L_2 \subseteq \{0,1\}^* $, we define $L_1 \simeq L_2$ as \\$  \exists m \forall s (s \in \{0,1\}^*) \rightarrow (\|s\| \ge m \rightarrow (s \in L_1 \leftrightarrow s \in L_2 ) )  $.
\end{defn}

It is easy to see that the ``$ \simeq $'' is an equivalence relation, and the following result is easily derived from the lemma \ref{lemma_004}.

\begin{corollary}
	\label{cor_002}
		Let $L_1 \subseteq \{0,1\}^*, L_2 \subseteq \{0,1\}^* $. \\If $L_1 \simeq L_2$ then  $L_1 \in \textbf{P} \leftrightarrow L_2 \in \textbf{P}$, and $L_1 \in \textbf{NP} \leftrightarrow L_2 \in \textbf{NP}$.
\end{corollary}

\begin{proof}
	Obviously.
\end{proof}

\begin{defn}
	\label{qms}
	Let $M_1$ is a Turing machine, and $w$ is an input string of $M_1$, $M_2$ is a Turing machine on $\{0,1\}^*$ and return $0$ or $1$, then $Q_{\langle M_1,M_2,w \rangle}$ is a Turing machine defined on $\{0,1\}^*$ as:\\$Q_{\langle M_1,M_2,w \rangle}$ = ``On input $s \in \{0,1\}^*$:
	
	\begin{enumerate}
		\item Use the machine $O_{\langle M_1, w \rangle}$ to compute on $s$.
		\item If $O_{\langle M_1, w \rangle}(s) = 1$, return 1, and halt.
		\item Else if $O_{\langle M_1, w \rangle}(s) = 0$, use the machine $M_2$ to compute on $s$ and return as $M_2$ return, and halt.''
	\end{enumerate}

	$L_{\langle M_1,M_2,w \rangle} = \{s\in\{0,1\}^*| Q_{\langle M_1,M_2,w \rangle}(s) = 1  \}$
\end{defn}

\begin{theorem}
	\label{pnp01}
	 There exists a language $L \subseteq \{0,1\}^*$, that the statement $L \in \textbf{P}$ is independent of ZFC.
\end{theorem}

\begin{proof}
	It is easy to see that there is a language  $L_1 \subseteq \{0,1\}^*, L_1 \notin \textbf{P}$. 
	
	Therefore we suppose $M_1$ is a Turing machine which decide the language $L_1$.
	
	Because $ L_1 \notin \textbf{P}$, $M_1$ is not in the polynomial time class.

	Let $\langle M_h, w_h \rangle$ is the machine and string pair as in the lemma \ref{lemma_003} that is:
	
	\begin{enumerate}
		\item $  (\forall s (s \in \{0,1\}^*) \rightarrow ( O_{\langle M_h, w_h \rangle}(s) = 1) ) $ cannot be proved in ZFC.
		\item $  (\exists m(\|s\| < m \rightarrow O_{\langle M_h, w_h \rangle}(s)=1 ) \wedge (\|s\| \ge m \rightarrow O_{\langle M_h, w_h \rangle}(s)=0) ) $ cannot be proved in ZFC.
		\item There exist a model of ZFC, $\mathcal{A}$, in which \\$  (\forall s (s \in \{0,1\}^*) \rightarrow ( O_{\langle M_h, w_h \rangle}(s) = 1) ) $ is true.
		\item There exist a model of ZFC, $\mathcal{B}$, in which \\$  (\exists m(\|s\| < m \rightarrow O_{\langle M_h, w_h \rangle}(s)=1 ) \wedge (\|s\| \ge m \rightarrow O_{\langle M_h, w_h \rangle}(s)=0) ) $ is true.
	\end{enumerate}

	Then we can define a Turing machine $M_2 = Q_{\langle M_h,M_1,w_h \rangle}$, that is:\\$M_2$ = ``On input $s \in \{0,1\}^*$:
	
	\begin{enumerate}
		\item Use the machine $O_{\langle M_h, w_h \rangle}$ to compute on $s$.
		\item If $O_{\langle M_h, w_h \rangle}(s) = 1$, return 1, and halt.
		\item Else if $O_{\langle M_h, w_h \rangle}(s) = 0$, use the machine $M_1$ to compute on $s$ and return as $M_1$ return, and halt.''
	\end{enumerate}
	
	Let $L = \{s\in\{0,1\}^*| M_2(s) = 1  \}$
	
	In model $\mathcal{A}$, because $ (\forall s (s \in \{0,1\}^*) \rightarrow ( O_{\langle M_h, w_h \rangle}(s) = 1) ) $ is true, the machine $M_2$ never execute the 3rd instruction as above definition, and from the lemma \ref{lemma_002}, $O_{\langle M_h, w_h \rangle} \in \textbf{P}$, therefore $L \in \textbf{P} $ is true in model $\mathcal{A}$.
	
	In model $\mathcal{B}$, because $  (\exists m(\|s\| < m \rightarrow O_{\langle M_h, w_h \rangle}(s)=1 ) \wedge (\|s\| \ge m \rightarrow O_{\langle M_h, w_h \rangle}(s)=0) ) $ is true, on the $s \in \{0,1\}^* \wedge \|s\| \ge m$ the machine $M_2$ must execute the 3rd instruction as above definition: use the machine $M_1$ to compute on $s$ and return as $M_1$ return. 
	
	That is $\|s\| \ge m \rightarrow M_2(s) = M_1(s)$. As suppose $M_1$ decide the language $L_1$, so $L \simeq L_1$. 
	
	Since the $L_1 \notin \textbf{P}$ as above supposed and the corollary \ref{cor_002}, $L \notin \textbf{P}$ is true in model $\mathcal{B}$.

	Hence  $L \in \textbf{P} $ is true in model $\mathcal{A}$, and $L \notin \textbf{P}$ is true in model $\mathcal{B}$. That means $L \in \textbf{P}$ is independent of ZFC.
	
\end{proof}

\begin{remark}
	The theorem \ref{pnp01} be proved by model method, indeed, it can be proved directly: if $L \in \textbf{P}$ can be proved in ZFC, then from the $L_1 \notin \textbf{P}$ and the definition of $M_2$, we can deduce $  (\forall s (s \in \{0,1\}^*) \rightarrow ( O_{\langle M_h, w_h \rangle}(s) = 1) ) $ in ZFC, contradicting the property of the selected $\langle M_h, w_h \rangle$, and if $L \notin \textbf{P}$ can be proved in ZFC, also result in a contradiction. Therefore $L \in \textbf{P}$ cannot be proved in ZFC and $L \notin \textbf{P}$ cannot be proved in ZFC either. That is $L \in \textbf{P}$ is independent of ZFC.
\end{remark}

\begin{remark}
	We can enumerate triplets $\langle M_1,M_2,w \rangle$ as the definition \ref{qms}. For some corresponding languages $L_{\langle M_1,M_2,w \rangle}$, the question: ``$L_{\langle M_1,M_2,w \rangle} \in \textbf{P} $ ?'' are easy, others are difficult. For example, let 
	
	$M_2$ is arbitrary non-$\textbf{P}$ class Turing machine, and 
	
	$M_1$ = ``On input $s \in \{0,1\}^*$:
	
	\begin{enumerate}
		\item If $\|s\|>1$, halt.
		\item If $\|s\|=1$, searching a even number $n$ satisfy $n>4$ and $n$ is not the sum of two odd primes, do not halt until found.''
	\end{enumerate}

Then the question: $L_{\langle M_1,M_2,1 \rangle} \in \textbf{P} $? is very difficult, indeed it is equivalent to Goldbach conjecture. But the theorem \ref{pnp01} says there exists language $L$, the question ``$L \in \textbf{P}$?'' is even more difficult that we cannot prove or disprove it, it is independent of ZFC.
\end{remark}

\begin{remark}
	Let  $L \in \textbf{P}$ is independent of ZFC. It is obvious that $(L \in \textbf{P}) \vee (L \notin \textbf{P})$ is true, but no one can decide whether $(L \in \textbf{P}) $ or $ (L \notin \textbf{P})$, the language $L$ is in a ``correlated'' state. More interesting quantum phenomena of computation will be study later, but now, another interesting theorem about $\textbf{P}$ versus $\textbf{NP}$ will be given in the following section.
\end{remark}

\begin{center}
	\item
	\section{A theorem about P versus NP problem}
\end{center}

In this section, a theorem about $ \textbf{P}\ \text{versus}\ \textbf{NP}$ problem be proved: there exists a language $L \in \textbf{NP}$, but for any polynomial time deterministic Turing machine $M$, we cannot prove $M$ decide $L$ in ZFC system.\\

There is an encoding system that can encode any $\langle M_1,M_2,P_1,P_2 \rangle$ as a natural number, where $M_1,M_2$ are Turing machines, $P_1,P_2$ are formula sequences, and there is a corresponding decode algorithm which can decode a number $n$ to a corresponding quadruplet $\langle M_1,M_2,P_1,P_2 \rangle$ if $n$ is an encoding number of it, else if $n$ is not an encoding number of any quadruplet $\langle M_1,M_2,P_1,P_2 \rangle$, return 0. Think of encryption and decryption algorithms in modern cryptography, this assumption is obviously true. We fixed one encoding system throughout this section, denote the encode algorithm by $T_{encode}$ and denote the corresponding decode algorithm by $T_{decode}$.\\

\begin{defn}
	\label{defn_MF}
	Let $M_0$ be a Turing machine on $\{0,1\}^*$ as in the \textbf{corollary} \ref{cor_of_lemma_03}. Now we define a language $L_{M_0},\ \text{such that }L_{M_0} \subset \omega$. Let \\
	$L_{M_0} = \{ n \in \omega |\ n = T_{encode}( \langle M_1,M_2,P_1,P_2 \rangle ), \ M_1, M_2 \text{ are Turing}\\
	\text{machines, and } P_1\text{ is a proof of }	M_1 \in \textbf{P} \text{ in ZFC, } P_2 \text{ is a proof of the following }\\
	\text{statement in ZFC: }M_2 \text{ compute a function: }\omega \to \omega \text{ and } M_0(M_2(n)) = M_1(n) \}$
	
	It is obvious that $L_{M_0}$ is a decidable language, let $H_{M_0}$ is a Turing machine which decide the language $L_{M_0}$. Now we define a Turing machine $T_{M_0}$ on $\omega$ as following:\\
	$T_{M_0}$ = ``On input $n$:
	\begin{enumerate}
		\item if $n=0$, return 0 and halts, else:
		\item use $H_{M_0}$ to decide whether $n \in L_{M_0} $,
		\item if $n \notin L_{M_0} $, return $1+T_{M_0}(n-1)$ and halts, else:
		\item if $n \in L_{M_0} $, use algorithm $T_{decode}$ to decode the number $n$ into\\
		$\langle M_1,M_2,P_1,P_2 \rangle$, and let $T_{M_0}$ return $\textbf{max}\{ 1+T_{M_0}(n-1),\ 1+M_2(n) \}$ and halts.
		
	\end{enumerate}
\end{defn}
Then the property of $T_{M_0}$ be described in following lemma:
\begin{lemma}
	\label{property_of_TM0}
	Let $M_0,T_{M_0}$ as the above \textbf{definition} \ref{defn_MF} and let the function $f(n) = M_0(T_{M_0}(n))$ then:
	\begin{enumerate}
		\item $T_{M_0}$ compute a function from $\omega$ to $\omega$,
		\item $T_{M_0}(n-1) < T_{M_0}(n) $ can be proved in ZFC,
		\item $\forall n \   (\  f(n)=0 \vee  f(n)=1  )$ can be proved in ZFC,
		\item $\forall m,n \   (m>n) \rightarrow ( f(m)=1 \rightarrow  f(n)=1  )$ can be proved in ZFC,
		\item $\forall m,n \   (m>n) \rightarrow ( f(n)=0 \rightarrow  f(m)=0  )$ can be proved in ZFC,
		\item $(\forall n  ( f(n) = 1)) \vee (\exists m(\ (n < m \rightarrow f(n)=1 ) \wedge (n \ge m \rightarrow f(n)=0)))$ can be proved in ZFC,
		\item the formula $\forall n  (f(n) = 1) $ is independent of ZFC,
		\item \label{p_t_08}for any Turing machine $M$, and any formula sequences $P_1 \text{ and } P_2$,\\
		$\text{if } n = T_{encode}( \langle M,T_{M_0},P_1,P_2 \rangle ), \text{ then } n \notin L_{M_0}$, that means if $M \in \textbf{P}$ can be proved in ZFC (i.e., the statement ``$M$ is a polynomial-time Turing machine'' can be proved in ZFC) then the statement: $$f(n)= M(n)$$ cannot be proved in ZFC.
		
	\end{enumerate}
\end{lemma}

\begin{proof}
	From the \textbf{definition} \ref{defn_MF} and the \textbf{corollary} \ref{cor_of_lemma_03}, it is easy to verify the statements from (1) to (7). Now we prove the statement (8) by contradiction.
	
	Assume that there exist a Turing machine $M$ and two formula sequences $P_1 \text{ and } P_2$ which satisfy $$ n = T_{encode}( \langle M,T_{M_0},P_1,P_2 \rangle ),\ n \in L_{M_0}$$ Then from the \textbf{definition} \ref{defn_MF}, $$T_{M_0}(n) =  \textbf{max}\{ 1+T_{M_0}(n-1),\ 1+T_{M_0}(n) \} \ge 1+T_{M_0}(n) > T_{M_0}(n) $$Thus we obtain a contradiction.
\end{proof}

\begin{theorem}
	\label{pvsnp}
	There exist a language $L$ which satisfies the following statements:
	\begin{enumerate}
		\item $L \in \textbf{NP} $,
		\item \label{pvsnp_con_2} For any Turing machine $T_M$, if we can prove the statement ``$T_M$ decide the language $L$'' in ZFC, then the statement ``$T_M$ is a polynomial-time Turing machine '' cannot be proved in ZFC, and vice versa.
	\end{enumerate}
\end{theorem}

\begin{proof}
	If $\textbf{P} \ne \textbf{NP} $, the theorem is obvious. 
	
	So we just need to show this theorem still holds under the assumption $\textbf{P} = \textbf{NP}$. Though $L \in \textbf{P} $ under this assumption, we cannot prove any concrete polynomial-time Turing machine decide it in ZFC. Indeed, despite whether or not $\textbf{P} = \textbf{NP}$ there exist a language $L$ satisfies this theorem. \\
	
	Let $L_0 \subset \omega$ be a language, $ L_0 \in \textbf{NP}$. Since the assumption $\textbf{P} = \textbf{NP}$, we get $\textbf{NP} = \textbf{coNP}$, therefore $\overline{L_0} \in \textbf{NP}$ where $\overline{L_0} = \omega \setminus L_0$. If $L_0$ satisfies the two statements of this theorem, the theorem be proved.

	Else let there is a Turing machine $T_{L_0}$, we can prove the statement ``$T_{L_0}$ decide the language $L_0$'' in ZFC, and $T_{L_0} \in \textbf{P}$ can be proved in ZFC. It is easy to see that $1-T_{L_0}$ decide $\overline{L_0}$.

	Let $f(n)$ be defined as in the \textbf{lemma} \ref{property_of_TM0}. Now we define a Turing machine $U$ on $\omega$ as following:\\
	$U$ = ``On input $n$:
	\begin{enumerate}
		\item First, compute $f(n)$,
		\item If $f(n)=1$, compute $T_{L_0}(n)$, return $T_{L_0}(n)$, and halts,
		\item Else if $f(n)=0$, compute $1- T_{L_0}(n)$, return $1- T_{L_0}(n)$, and halts.
	\end{enumerate}
	Let $L_U$ denote the language $\{n|\ U(n)=1\}$. It is not hard to see that $$U(n) = T_{L_0}(n) \text{ if and only if } f(n)=1$$ thus we can prove the following formula in ZFC:
	
	\begin{equation}
		\label{U(n)=TL0(n)}
		U(n) = T_{L_0}(n) \leftrightarrow f(n)=1
	\end{equation}
	From the \textbf{lemma} \ref{property_of_TM0}, $f(n)$ satisfies: $$(\forall n  ( f(n) = 1)) \vee (\exists m(\ (n < m \rightarrow f(n)=1 ) \wedge (n \ge m \rightarrow f(n)=0)))$$Therefore from the definition of $U$ we can prove the following formula in ZFC:$$(\forall n  ( U(n) = T_{L_0}(n))) \vee (\exists m(\ (n < m \rightarrow U(n)=T_{L_0}(n) ) \wedge (n \ge m \rightarrow U(n)=1-T_{L_0}(n))))$$From the \textbf{definition} \ref{apeq} and above formula, we get $(L_U \simeq L_0) \vee (L_U \simeq \overline{L_0} ) $. Under the assumption $\textbf{P} = \textbf{NP}$, we know that  $ L_0 \in \textbf{NP}$ and $\overline{L_0} \in \textbf{NP}$. Since the \textbf{corollary} \ref{cor_002}, the formula $L_U \in \textbf{NP}$ can be proved in ZFC. Though $L_U \in \textbf{P}$ under the assumption $\textbf{P} = \textbf{NP}$, we now prove $L_U$ satisfies the statement \ref{pvsnp_con_2} of this theorem by contradiction.
	
	Assume that there exist a concrete Turing machine $T_M$ that we can prove the statement ``$T_M$ decide the language $L_U$'' in ZFC and ``$T_M$ is a polynomial-time Turing machine '' can be proved in ZFC, too, i.e., we assumed the following two formulas be proved in ZFC:
	
	\begin{equation}
		\label{U(n)=TM(n)}
		\forall n (U(n) = T_M(n))
	\end{equation}
	
	\begin{equation}
		T_M \in \textbf{P}
	\end{equation}
	Now, we define a polynomial-time Turing machine $M$ compute $f(n)$, which will contradict the statement \ref{p_t_08} of the \textbf{lemma} \ref{property_of_TM0}:\\
	
	$M$ = ``On input $n$:
	\begin{enumerate}
		\item Compute $T_M(n)$, (using polynomial-time as assumed)
		\item Compute $T_{L_0}(n)$, (using polynomial-time as assumed)
		\item Compare $T_M(n)$ to $T_{L_0}(n)$, if they are equal, return 1, and halts,
		\item Else, return 0, and halts.
	\end{enumerate}
	From the definition of $M$,  and $T_M \in \textbf{P}$, $T_{L_0} \in \textbf{P}$ are all proved in ZFC as assumptions, it is easy to see the following formulas can be proved in ZFC:

	\begin{equation}
		\label{MinP}
		M \in \textbf{P}
	\end{equation}
	
	\begin{equation}
		\label{Un=0or1}
		\forall n \   (\  M(n)=0 \vee  M(n)=1  )
	\end{equation}

	\begin{equation}
		\label{Mn=1}
		M(n)=1 \leftrightarrow T_M(n) = T_{L_0}(n)
	\end{equation}
	Therefore from the formula (\ref{U(n)=TM(n)}) and formula (\ref{Mn=1}), the following formula can be proved in ZFC:
	
	\begin{equation}
		\label{MntoUn}
		M(n)=1 \leftrightarrow U(n) = T_{L_0}(n)
	\end{equation}
	Then from the formula (\ref{U(n)=TL0(n)}) and formula (\ref{MntoUn}), the following formula can be proved in ZFC:
	
	\begin{equation}
		\label{Mn=1iffn=1}
		M(n)=1 \leftrightarrow f(n)=1
	\end{equation}
	Hence from the property of $f(n)$ in the \textbf{lemma} \ref{property_of_TM0}, formula (\ref{Un=0or1}) and (\ref{Mn=1iffn=1}) the following formula can be proved in ZFC:
	
	\begin{equation}
		\label{Mn=fn}
		f(n)=M(n)
	\end{equation}
	Therefore, we can prove the formula (\ref{MinP}): $M \in \textbf{P}$ in ZFC, and the formula (\ref{Mn=fn}): $f(n)=M(n)$ can be proved in ZFC, too. That violates the \ref{p_t_08}th statement of the \textbf{lemma} \ref{property_of_TM0} and is thus a contradiction.
	
	So the assumption: the two statements ``$T_M$ decide the language $L_U$'' and ``$T_M$ is a polynomial-time Turing machine '' can be proved all in ZFC, is wrong. Thus we proved $L_U$ satisfies the statement \ref{pvsnp_con_2} of this theorem.

\end{proof}

\end{document}